\documentclass[10pt,journal,compsoc]{IEEEtran}
\usepackage{amsmath,amsthm,amssymb,paralist,subfigure,cite,graphicx,color,amsbsy,float}
\usepackage{stmaryrd,mathrsfs,url}
\usepackage[table]{xcolor}
\usepackage{cuted,mathtools,lipsum}
\usepackage[ruled,vlined,linesnumbered]{algorithm2e}
\usepackage{soul}

\newtheorem{lem}{Lemma}
\newtheorem{theorem}{Theorem}

\newtheorem{rem}{Remark}
\newtheorem{ass}{Assumption}

\def\mb{\mathbf}

\def\mc{\mathcal}

\DeclareMathOperator*{\argmin}{argmin}

\begin{document}
\title{Nonlinear Perturbation-based Non-Convex Optimization over Time-Varying Networks}

\author{Mohammadreza~Doostmohammadian, Zulfiya~R.~Gabidullina, and Hamid~R.~Rabiee,~\IEEEmembership{Senior~Member,~IEEE}
\IEEEcompsocitemizethanks{\IEEEcompsocthanksitem M.~Doostmohammadian is with the
Faculty of Mechanical Engineering, Semnan University, Semnan, Iran, and with Centre for International Scientific Studies and Collaborations, Tehran, Iran. E-mail: doost@semnan.ac.ir
\IEEEcompsocthanksitem Z. R. Gabidullina is with the Institute of Computational Mathematics and Information Technologies, Kazan Federal University, Russia, E-mail: Zulfiya.Gabidullina@kpfu.ru.
\IEEEcompsocthanksitem H. R. Rabiee is with the School of Computer Engineering, Sharif University of Technology, Tehran, Iran. E-mail: rabiee@sharif.edu
}
}

\markboth{IEEE Transaction on Network Science and Engineering}%
{Doostmohammadian \MakeLowercase{\textit{et al.}}: Nonlinear Perturbation-based Non-Convex Optimization over Time-Varying Networks}

\IEEEtitleabstractindextext{
\begin{abstract}
Decentralized optimization strategies are helpful for various applications, from networked estimation to distributed machine learning. This paper studies finite-sum minimization
problems described over a network of nodes and proposes a computationally efficient algorithm that solves distributed convex problems and optimally finds the solution to locally non-convex objective functions. In contrast to batch gradient optimization in some literature, our algorithm is on a single-time scale with no extra inner consensus loop. It evaluates one gradient entry per node per time. Further, the algorithm addresses link-level nonlinearity representing, for example, logarithmic quantization of the exchanged data or clipping of the exchanged data bits. Leveraging perturbation-based theory and algebraic Laplacian network analysis proves optimal convergence and dynamics stability over time-varying and switching networks. The time-varying network setup might be due to packet drops or link failures. Despite the nonlinear nature of the dynamics, we prove exact convergence in the face of odd sign-preserving sector-bound nonlinear data transmission over the links. Illustrative numerical simulations further highlight our contributions.
\end{abstract}
\begin{IEEEkeywords}
Distributed Algorithm, Locally Non-Convex Optimization, Network Science, Perturbation Theory
\end{IEEEkeywords}}
\maketitle
\IEEEdisplaynontitleabstractindextext
\IEEEpeerreviewmaketitle

\IEEEraisesectionheading{\section{Introduction}\label{sec_intro}}
\IEEEPARstart{D}{ata}
Many real-world systems are generically distributed over large multi-agent networks (or sensor networks). This, along with recent advances in parallel data processing and cloud computing, motivates decentralized data mining and machine learning solutions \cite{cao2021survey}. In contrast to centralized optimization techniques \cite{danilova2022recent}, the parallelized and distributed counterparts benefit from no-single-node-of-failure and faster parallel data processing over multiple nodes. Recent applications of interest include distributed estimation \cite{tnse_19}, reinforcement learning \cite{dai2020distributed}, data classification via support-vector-machine \cite{dsvm}, robotic network task assignment \cite{luo2015distributed,MiadCons}, federated learning \cite{li2023analysis}, distributed energy resource management \cite{doostmohammadian2023distributed,khojasteh2023distributed}, and linear/logistic regression \cite{qureshi2023distributed} among others. The interplay between network dynamics, agent interactions, and nonlinearity in communication links presents a rich and challenging research direction in these setups. These distributed strategies must address the real-world constraints of the communication and data exchange networks. For example, the exchanged data among the nodes might be subject to non-ideal (or nonlinear) constraints such as quantization and saturation (or clipping). Further, mobile robotic and multi-agent networks are generically dynamic, with nodes coming into and out of each other's communication (and broadcasting) range. This mandates dynamic and switching network consideration for parallelization of the information processing. More recently, non-convex optimization models \cite{danilova2022recent} are considered with applications in binary classification and regularization techniques. Moreover, the objective functions in many control applications are locally non-convex and require the proper design of the networked optimization dynamics. This work focuses on decentralized techniques to optimize (possibly) non-convex local objective functions over dynamic networks subject to link nonlinearity.

\textit{Literature review:}
First-order methods based on gradient descent are primarily proposed in \cite{nedic2009distributed,kar2012distributed,chen2012diffusion,mateos2014distributed,dsvm} for convex problems. Distributed and decentralized optimization in the presence of malicious agents \cite{zhang2023accelerated}, coupled with a feasibility constraint \cite{doostmohammadian2024accelerated,alghunaim2019distributed,ojsys,lian2023distributed,scl}, subject to uniform data quantization \cite{rikos2023distributed}, with fixed-time convergence \cite{garg2020fixed1,garg2020fixed}, and over dynamic balanced digraphs subject to link failure \cite{ddsvm} are considered in the literature.
Most of these existing works focus on \textit{convex} models. Few analytical works study non-convex finite-sum models \cite{xin2021fast,tatarenko2017non,vlaski2021distributed,kao2023localization} and min-max optimization (or saddle-point problems) \cite{qureshi2023distributed,mokhtari2020unified,yang2022faster} aiming to maximize the objective
in one direction and minimize in the other. Other existing non-convex optimization solutions include primal-dual \cite{scutari2016parallel,hajinezhad2016nestt,barazandeh2021decentralized} and gradient-free (or zeroth-order) algorithms \cite{yi2022zeroth,bogunovic2018adversarially,liu2020primer,maheshwari2022zeroth,zhang2022accelerated}. Some other literature focus on double-timescale (with inner consensus loop) solutions \cite{xin2022fast,sun2020improving},
mandating periodic \textit{batch} gradient evaluations along with component gradient computations at each iteration, and thus, suffer from periodic network synchronizations. The gap in the current non-convex optimization literature is on lack of considering non-ideal (or nonlinear) constraints such as logarithmic quantization or clipping while addressing the time-varying (dynamic) nature of the underlying multi-agent network in the case of link failures or information loss.

\textit{Contributions:}
In this work, we study non-convex distributed optimization problems by adopting gradient tracking (GT) techniques and introducing an auxiliary variable. The proposed algorithm can tune the agreement rate (consensus) on the state values versus the GT step rate via specific parameters. The algorithm is on a single-timescale and, thus, is computationally efficient in contrast to the double-timescale scenarios. We further consider non-linear (for example, logarithmically quantized or clipped) data exchange over the links connecting the nodes. This is to account for possible non-ideal linking conditions in real-world distributed setups. This proposed methodology addresses many nonlinearities over the network that satisfy certain odd sign-preserving sector-bound assumptions. Due to this nonlinear model, we combine matrix perturbation theory and graph spectral analysis to prove convergence. This further allows one to address convergence over possibly time-varying and switching network topologies among the agents. This is the case, for example, in mobile sensor networks and swarm robotics, where dynamic network consideration is a must. Moreover, this proposed technique addresses convergence in link failure (e.g., because of information loss or packet drops). We perform extensive simulations to verify our results over optimization and learning setups. To the best of our knowledge, the intersection of optimization, dynamic networks, and consensus dynamics capable of addressing the complexities posed by data transmission subject to log-quantization (and general sector-bound nonlinearity) is not considered in the literature. This work opens possible new directions in the machine learning and data mining scenarios addressing real-world nonlinear constraints in contrast to the existing literature.

\textit{Paper organization:} Section~\ref{sec_prob} formulates the finite-sum non-convex optimization problem. Section~\ref{sec_alg} presents our distributed optimization algorithm. Section~\ref{sec_conv} presents the perturbation-based convergence analysis. Section~\ref{sec_sim} provides extensive illustrative simulations, and Section~\ref{sec_con} concludes the paper.

\textit{Notations:} Vectors $\mb{1}_n,\mb{0}_n$ denote the all $1,0$ column vectors of size $n$ respectively. $I_n$ is the identity matrix of size $n$. Operator $\otimes$ denotes the Kronecker product of matrices. $\nabla f(\mb{x})$ and $\nabla^2 f(\mb{x})$ are the gradient and second derivative of $f(\mb{x})$ with respect to $\mb{x}$. Operator $\partial_t$ denotes $\frac{d}{dt}$. $\lVert A\rVert_{\infty}$ denotes the infinity norm of matrix $A=[a_{ij}]$ defined as $\lVert A\rVert_{\infty} = \max_{1\leq i\leq n} \sum_{j=1}^n |a_{ij}|$. RHP and LHP stand for right-half-plane and left-half-plane, respectively. Operator ``$\succ$'' denotes matrix positive definiteness.

\section{Optimization Problem Formulation}\label{sec_prob}
The problem is described over a network $\mc{G}_\gamma=\{\mc{V},\mc{E}_\gamma\}$ with set of nodes/agents $\mc{V}$ of size $n$ and communications over the link set $\mc{E}_\gamma$. The \textit{switching} signal $\gamma: t \mapsto \Gamma$ assigns the network topology at time $t$ from a finite set $\Gamma$. A link $(a,b) \in \mc{E}_\gamma$ implies a data-exchange channel from node $a$ to node $b$ over the time-varying graph $\mc{G}_\gamma$; this defines the neighboring set as $a \in \mc{N}_b$. Every node $i$ have a set of local smooth objective functions $\{f_{i,j}: \mathbb{R}^p \mapsto \mathbb{R}\}_{j=1}^m$ that might be \textit{non-convex}. This setting is motivated by a data analysis framework where the $j$th data sample incurs the $j$th cost function. Then, the problem is to optimize the following objective function,
\begin{align}
\min_{\mb{x} \in \mathbb{R}^{p}} &
F(\mb x) = \frac{1}{n}\sum_{i=1}^{n} f_i(\mb{x}), \label{eq_prob0}
\end{align}
with $\mb x$ as the global state variable.
Note that the coupling between the nodes can be moved from the objective to the \textit{consensus-consistency constraint}, i.e., in a distributed setup, the problem can be reformulated as \cite{dsvm,rikos2023distributed},
\begin{align}\nonumber
\min_{\mb{x} \in \mathbb{R}^{p}} &
F(\mb x) = \frac{1}{n}\sum_{i=1}^{n} f_i(\mb{x}_i)\\\label{eq_prob}
\text{subject to}& ~ \mb{x}_1 = \mb{x}_2 = \dots = \mb{x}_n,
\end{align}
where the global column vector is $\mb{x}=[\mb{x}_1;\mb{x}_2;\dots;\mb{x}_n]$ (with operator ``;'' as the column concatenation), the parameter $\mb{x}_i$ is the state at node $i$ (representing the local state assessment at node $i$ that must reach consensus on), and the local objectives are in the form
\begin{align}\label{eq_fij}
f_i(\mb{x}_i) = \frac{1}{m}\sum_{j=1}^{m} f_{i,j}(\mb{x}_i).
\end{align}
This cost model implies that the nodes agree on a stationary state of the average of all constituent local functions via localized information processing and communications. An example application is given in Section~\ref{sec_sim}, in which the local state parameter $\mb{x}_i$ denotes the regressor line parameter at node $i$ based on its local data. 

Define $H:=\mbox{diag}[\nabla^2 f_i(\mb{x}_i)]$ as the Hessian matrix. \textit{ The local cost functions $f_i(\mb{x}_i)$ might be non-convex, i.e., $\nabla^2 f_i(\mb x_i)$ is not necessarily positive.} This implies that matrix $H$ \textit{is not necessarily positive semi-definite (PSD)}.

\begin{ass} \label{ass_cost}
The global cost function is assumed to satisfy the following,
\begin{align}\label{eq_H1}
(\mb{1}_n \otimes I_p)^\top H (\mb{1}_n \otimes I_p) \succ 0.
\end{align}
It is further assumed that all the local costs are differentiable and Lipschitz. The global objective satisfies $F^* := \inf_{\mb{x} \in \mathbb{R}^{p}} F(\mb{x}) > -\infty$ with $F^*$ as the optimal value.
\end{ass}
\begin{rem}
The Assumption~\ref{ass_cost} implies that, although the local cost functions $f_i(\cdot)$ and the global cost function $F(\cdot)$ might be non-convex, the global cost $F(\cdot)$ has no local minimum and it only has a global minimum. However, the local cost functions $f_i(\cdot)$ may have a local minimum (an example is given later in Fig.~\ref{fig_nonconv} in Section~\ref{sec_sim_nonconv}). Similar assumption is made in some recent works \cite{xin2021improved,xin2021fast,fazel2018global,karimi2016linear,zhao2022beer,kao2024localization,mancino2023decentralized,lin2024stochastic}.
\end{rem}
\begin{ass}
The multi-agent network with adjacency matrix $W_\gamma$ is assumed to be directed, strongly connected, and weight-balanced (WB), i.e., $\sum_{i=1}^n w^\gamma_{ij}= \sum_{j=1}^n w^\gamma_{ij}$. We consider non-negative entries satisfying
\begin{align} \label{eq_sum_wij}
\sum_{j=1}^n w_{ij}^\gamma<1.
\end{align}
\end{ass}
The Laplacian matrix associated with the graph topology is represented by $\overline{W}_\gamma=\{\overline{w}^\gamma_{ij}\}$ with entries defined as $\overline{w}^\gamma_{ij}=w^\gamma_{ij}$ for $i\neq j$ and $\overline{w}^\gamma_{ij}=-\sum_{i=1}^n w^\gamma_{ij}$ for $i=j$.
It is known from \cite{SensNets:Olfati04,olfatisaberfaxmurray07} that for a strongly connected WB graph, the zero eigenvalue of $\overline{W}_\gamma$ is simple with left (and right) eigenvector $\mb{1}_n^\top$ (and~$\mb{1}_n$), i.e., $\mb{1}_n^\top \overline{W}_\gamma= \mb{0}_n$ and~$\overline{W}_\gamma \mb{1}_n=\mb{0}_n$. The other non-zero eigenvalues of $\overline{W}_\gamma$ are all in the LHP. Note that for undirected $\mc{G}_\gamma$, all the eigenvalues of the matrix $\overline{W}_\gamma$ are real; one can generalize this by considering \textit{mirror} graphs for strongly-connected \textit{directed} networks.

\section{The Proposed Algorithm over Switching Networks}\label{sec_alg}
To solve the distributed problem~\eqref{eq_prob}-\eqref{eq_fij} we propose a nonlinear perturbation-based gradient-tracking (\textbf{NP-GT}) solution as follows:
\begin{align} \label{eq_xdot_g} 
\dot{\mb{x}}_i &= -\sum_{j=1}^{n} w_{ij}^\gamma (h_l(\mb{x}_i)-h_l(\mb{x}_j))-\eta \mb{y}_i, \\ \label{eq_ydot_g}
\dot{\mb{y}}_i &= -\sum_{j=1}^{n} w_{ij}^\gamma (h_l(\mb{y}_i)-h_l(\mb{y}_j) ) + \partial_t \nabla f_i(\mb{x}_i),
\end{align}
With $\eta$ as GT step-rate, which denotes the rate at which the proposed solution tracks the average gradient at neighboring nodes, $\mb{y}_i$ as the auxiliary GT variable at node $i$, $h_l:\mathbb{R}^p \mapsto \mathbb{R}^p$ as the nonlinear function affecting the information sent over the links. One main difference with the existing GT-based solution is due to this nonlinearity consideration, which allows the addressing of non-ideal data exchange over the links due to, for example, quantization or saturation. The other feature of the proposed solution, as compared to other existing GT-based and ADMM-based \cite{cdc_dtac,mota2013d,mancino2023decentralized} solutions, is its resilience to change in the network topology. The solution is summarized in Algorithm~\ref{alg_1}.
\begin{algorithm} \label{alg_1}
\textbf{Given:} $f_{i,j}(\mb{x}_i)$, $\mc{G}_\gamma$, $W_\gamma$, $\eta$ \\ 
\textbf{Initialization:} ${\mb{y}}_i(0)=\mb{0}_{p}$, random ${\mb{x}}_i(0)$
\\
\While{termination criteria NOT true}{
Node $i$ finds local gradient $\boldsymbol{ \nabla} f_i(\mb{x}_i)$ \;
Node $i$ receives $\mb{x}_j$ and $\mb{y}_j$ from $j \in \mc{N}_i$ \;
Node $i$ carries out the calculations by Eqs.~\eqref{eq_xdot_g}-\eqref{eq_ydot_g} \;
Node $i$ shares updated $\mb{x}_i$ and $\mb{y}_i$ over $\mc{G}_\gamma$ \;
}
\textbf{Return:} optimal $\mb{x}^*$ and $F^*$\; 
\caption{\textbf{NP-GT} at node $i$. }
\end{algorithm}

One key point of the proposed algorithm is that it only needs WB networks instead of weight-stochastic networks. WB condition is known to be milder and more easily satisfied than row/column-stochastic condition. This is particularly important in the case of a link failure, as illustrated in Fig.~\ref{fig_linkremov}. For weight-stochastic-based distributed optimization works \cite{xin2021fast,tatarenko2017non,vlaski2021distributed,kao2023localization,xin2021improved,scutari2016parallel,hajinezhad2016nestt,barazandeh2021decentralized,xin2022fast,sun2020improving},
it is required to redesign the weights to satisfy stochasticity after link failure (or any change in the network topology). This is done via weight compensation algorithms introduced in \cite{6426252,cons_drop_siam} that add more complexity to the existing distributed optimization literature. However, satisfying the WB condition is much easier after a link failure without any complicated algorithm. This feature makes our solution more resilient to changes in the network topology and possible link failure. The other key feature in this regard is the perturbation-based convergence analysis discussed in detail later in Section~\ref{sec_conv}. 
\begin{figure} [hbpt]
\centering
\includegraphics[width=1.5in]{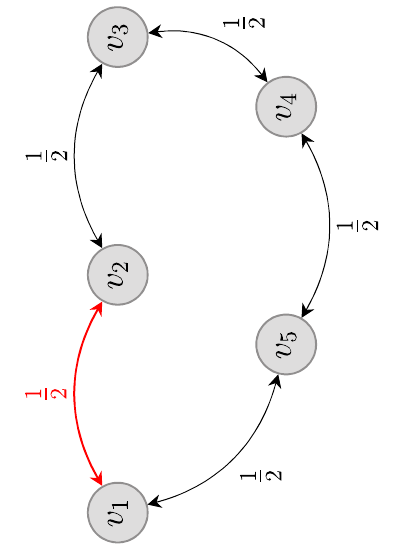}
\includegraphics[width=1.5in]{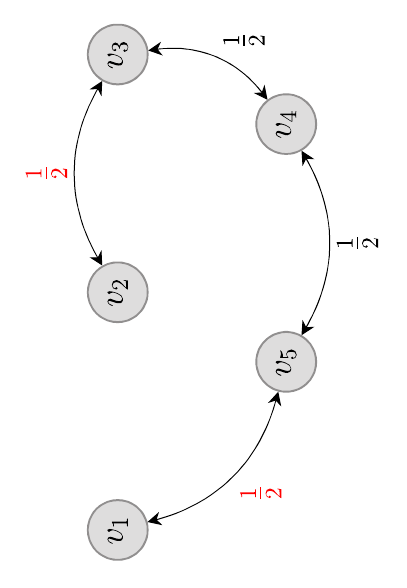} 
\caption{The left figure shows an undirected network (stochastic and WB) with an unreliable red-colored link. After the red link fails, the resulting network is not stochastic anymore but is still WB. This shows that the WB condition is milder than the stochastic condition. Therefore, our optimization solution converges over such unreliable networks, while much of the existing literature does not converge without redesigning the stochastic weights.}
\label{fig_linkremov}
\end{figure}

The function $h_l(\cdot)$ is odd, sign-preserving, and monotonically non-decreasing. Define the sector-bounds for such nonlinear mapping $h_l$ as $0<\kappa \leq \frac{h_l(z)}{z} \leq \mc{K}$. These bounds imply that the nonlinear function $h_l(z)$ is lower and upper-bounded by lines $\kappa z$ and $\mc{K} z$. One example of such a nonlinear mapping is \textit{log-scale quantization} defined as
$$ h_l(z) = \mbox{sgn}(z)\exp\left(\rho\left[\dfrac{\log(|z|)}{\rho}\right] \right),$$
where $[\cdot]$ denotes rounding to the nearest integer, and $\mbox{sgn}(\cdot)$, $\mbox{exp}(\cdot)$, and $\mbox{log}(\cdot)$ are the sign, exponential, and logarithmic functions, respectively. With $\rho$ as the logarithmic quantization level, we have $\kappa = 1-\dfrac{\rho}{2} \leq \dfrac{h_l(z)}{z}\leq 1+\dfrac{\rho}{2} = \mc{K}$. This implies that the log-quantization satisfies the assumption of sector-bound nonlinear mapping. This is better illustrated by Fig.~\ref{fig_logquant}.
\begin{figure}
\centering
\includegraphics[width=2.25in]{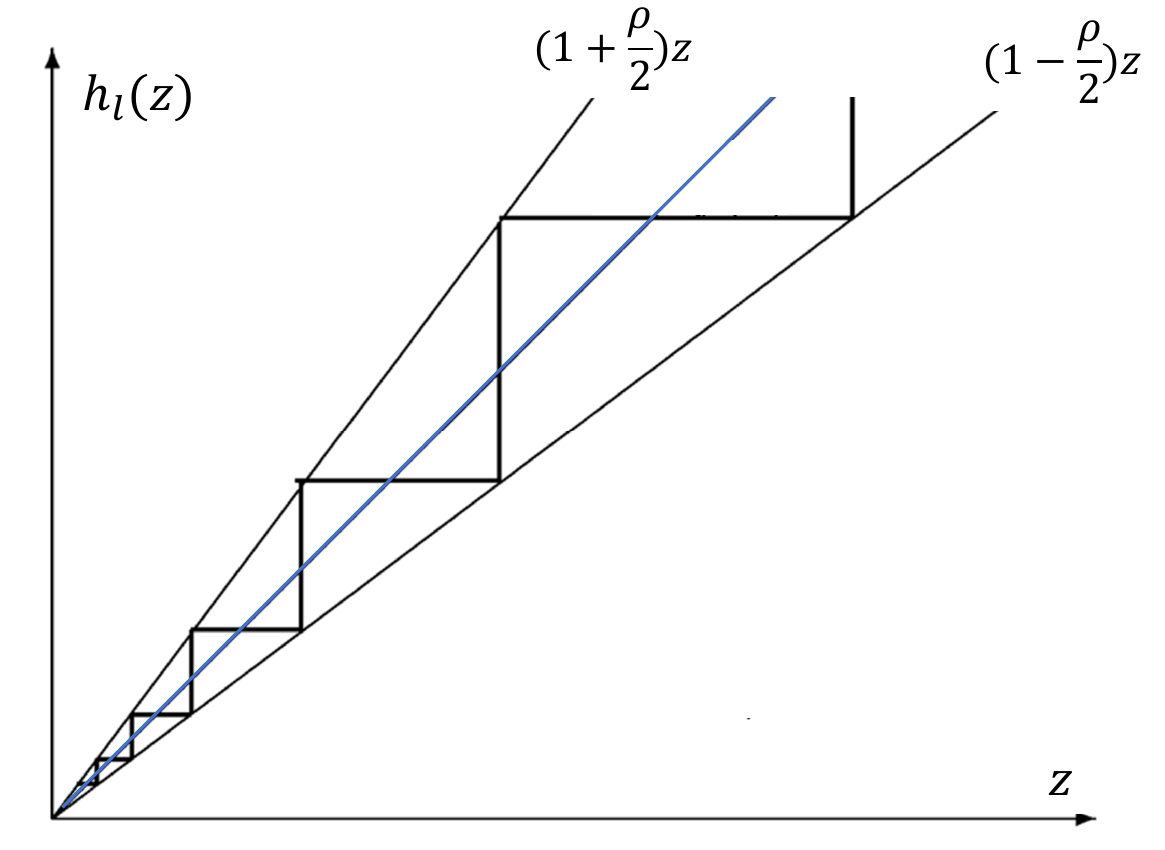} 
\caption{This figure shows the logarithmic quantization as an example of nonlinear mapping satisfying the sector-bound condition. For the quantization level $\rho$ the lines $1 \pm \frac{\rho}{2}$ bound the nonlinear function. } \label{fig_logquant}
\end{figure}
Further, note that the WB assumption ensures that under~\eqref{eq_xdot_g} and~\eqref{eq_ydot_g}, it is fulfilled:
\begin{align} \label{eq_sumydot}
\sum_{i=1}^n \dot{\mb{y}}_i
= \sum_{i=1}^n \partial_t \nabla f_i(\mb{x}_i), ~~
\sum_{i=1}^n \dot{\mb{x}}_i
= -\eta \sum_{i=1}^n\mb{y}_i.
\end{align}
Setting the initial condition as~$\mb{y}(0)=\mb{0}_{np}$, the GT nature of the dynamics straightly follows as
\begin{eqnarray} \label{eq_sumxdot2}
\sum_{i=1}^n \dot{\mb{x}}_i = -\eta \sum_{i=1}^n\mb{y}_i = -\eta \sum_{i=1}^n \nabla f_i(\mb{x}_i),
\end{eqnarray}
The consensus-type GT-based dynamics allow for considering odd sign-preserving model nonlinearities without violating the GT property. This particularly advances the recent primal-dual formulations \cite{scutari2016parallel,hajinezhad2016nestt,barazandeh2021decentralized} with no such considerations of real-world nonlinearity.
Then, for any random initial values~$\mb{x}(0) \notin \mbox{span} \{\mb{1}_{np}\}$ and~$\mb{y}(0)=\mb{0}_{np}$, at the optimizer~${\mb{x}=\mb{x}^*=\mb{1}_n \otimes \overline{ \mb{x}}^*}$, we have $\sum_{i=1}^n \dot{\mb{x}}_i = -\eta (\mathbf 1_n^\top \otimes I_p) \nabla F(\mb{x}^*) = \mb{0}_p$. Recalling the proposed dynamics~\eqref{eq_xdot_g}-\eqref{eq_ydot_g}, the optimizer further satisfies $\dot{\mb{x}}^*_i = \mb{0}_p$ and $\dot{\mb{y}}_i^* = \partial_t \nabla f_i(\overline{ \mb{x}}^*) = \nabla^2 f_i(\overline{ \mb{x}}^*) \dot{\mb{x}}^*_i = \mb{0}_p$. This implies that the point $[\mb{x}^*;\mb{0}_{np}]$ is the (invariant) equilibria of Eqs.~\eqref{eq_xdot_g}-\eqref{eq_ydot_g}.

One can linearize Eqs.~\eqref{eq_xdot_g}-\eqref{eq_ydot_g} at every time $t$. Then, recalling the sector-bound condition on the nonlinearity, the \textbf{NP-GT} dynamics in compact formulation can be rewritten as
\begin{align} \label{eq_xydot1}
\left(\begin{array}{c} \dot{\mb{x}} \\ \dot{\mb{y}} \end{array} \right) &= A_h(t,\eta,\gamma) \left(\begin{array}{c} {\mb{x}} \\ {\mb{y}} \end{array} \right), \\ \label{eq_M_g}
A_h(t,\eta,\gamma) &= \left(\begin{array}{cc} \overline{W}_{\gamma,\Xi} \otimes I_p & -\eta I_{np} \\ H(\overline{W}_{\gamma,\Xi}\otimes I_p) & \overline{W}_{\gamma,\Xi} \otimes I_p - \eta H
\end{array} \right),
\end{align}
with time-varying compact matrix $A_h(t,\eta,\gamma)$ as a function of $\eta$, switching signal $\gamma$, and $\Xi$ as the linearization of $h_l(\cdot)$ at the operating point $t$.
Setting $h_l(z) = z$ (i.e., the linear case), the dynamics matrix can be written as $A := A^0 + \eta A^1$ with
\begin{eqnarray}\nonumber
A^0 &=& \left(\begin{array}{cc} \overline{W}_\gamma \otimes I_p & \mb{0}_{np\times np} \\ H(\overline{W}_\gamma \otimes I_p) & \overline{W}_\gamma \otimes I_p \end{array} \right),\\\nonumber
A^1 &=& \left(\begin{array}{cc} \mb{0}_{np \times np} & - {I_{np}} \\ {\mb{0}_{np\times np}} & - H \end{array} \right).
\end{eqnarray}
Considering the nonlinearity at the links (i.e., setting $h_l(z) \neq z$), the following arguments hold for the compact forms
\eqref{eq_xydot1}-\eqref{eq_M_g}:
\begin{align} \label{eq_Mg}
A_h(t,\eta,\gamma) &= A_h^0 + \eta A^1, \\ \label{eq_beta_M0}
\kappa A^0 & \preceq A_h^0 \preceq \mc{K} A^0, \\ \label{eq_beta_M}
A_h^0 = \Xi(t) A^0 &,~ \kappa I_n \preceq \Xi(t) \preceq \mc{K} I_n,
\end{align}
where $A_h^0$ is the counterpart of $A^0$ for the nonlinear dynamics, $\Xi(t) := \mbox{diag}[\xi(t)]$, column vector $\xi(t) = [\xi_1(t);\xi_2(t);\dots;\xi_n(t)]$ with $\xi_i(t) = \frac{h_l(\mb{x}_i)}{\mb{x}_i}$ (or one can write $h_l(\mb{x}(t)) = \Xi(t) \mb{x}(t)$). Recall that the considered nonlinear mapping $h_l(\cdot)$ satisfies $\kappa \leq \xi_i(t) \leq \mc{K}$; this implies that the following holds:
\begin{align}
\overline{W}_{\gamma,\Xi} = \overline{W}_{\gamma} \Xi(t).
\end{align}
Therefore, one can relate the eigen-spectrum of the linear case, denoted by $\sigma(\overline{W}_{\gamma})$, to that of the nonlinear case, denoted by $\sigma(\overline{W}_{\gamma,\Xi})$. Following from the eigen-spectrum definition and the fact that $\Xi(t)$ is a diagonal matrix, one can write the determinant formulation as
\begin{align}
\mbox{det}(\overline{W}_{\gamma,\Xi}-\lambda I_{np}) = \mbox{det}(\overline{W}_{\gamma}-\lambda \Xi(t)^{-1}),
\end{align}
where $\lambda$ denotes the eigenvalue. Therefore, from \eqref{eq_beta_M0}-\eqref{eq_beta_M} we have
\begin{align} \label{eq_spect_k}
\kappa \sigma(A^0) \leq \sigma(A^0_h) \leq \mc{K} \sigma(A^0),
\end{align}
where $\sigma(A^0) = \sigma(\overline{W} \otimes I_p) \cup \sigma(\overline{W} \otimes I_p)$. This implies that the eigenspectrum of matrix $A^0$ includes two sets of the eigenspectrum of $\overline{W} \otimes I_p$. In case the two consensus matrix used in the dynamics \eqref{eq_xdot_g} and \eqref{eq_ydot_g} are different, for example as $W_1$ and $W_2$ matrices, we have $\sigma(A^0) = \sigma(\overline{W}_1 \otimes I_p) \cup \sigma(\overline{W}_2 \otimes I_p)$. These notions are used in the perturbation-based proof analysis in the next section.

\section{Perturbation-based Proof of Convergence}\label{sec_conv}
In the section, we prove the convergence of the Algorithm~\ref{alg_1} to the optimal point of problem \eqref{eq_prob}-\eqref{eq_fij}. For notation simplicity, we drop $(t,\eta,\gamma)$ unless where it is needed. First, we recall some perturbation-based theory results.

\begin{lem} \label{lem_pert}
\cite{stewart_book,cai2012average} Consider matrix $A(\eta)$ of size $n$ which smoothly depends on variable $\eta \geq 0$. Let $l \in \{1,\dots,n\}$
and $\lambda_1,\dots,\lambda_l$ be semi-simple eigenvalues of matrix $A^0$, with (linearly independent) right and left eigenvectors $\mb{v}_1,\dots,\mb{v}_l$ and $\mb{u}_1,\dots,\mb{u}_l$ such that
$$ [\mb{v}_1,\dots,\mb{v}_l]^\top [\mb{u}_1,\dots,\mb{u}_l] = I_l.
$$
Let $\lambda_i(\eta)$ be the $i$th eigenvalue of $A(\eta)$ corresponding to $\lambda_i$ as the $i$th eigenvalue of $A^0$. Then, $\partial_\eta \lambda_i(\eta)|_{\eta=0}$ is the $i$th eigenvalue $\lambda^\mc{S}_i$ of the following matrix of size $l$,
\begin{align}\left(\begin{array}{ccc}
\mb{u}_1^\top A' \mb{v}_1 & \ldots & \mb{u}_1^\top A' \mb{v}_l \\
& \ddots & \\
\mb{u}_l^\top A' \mb{v}_1 & \ldots & \mb{u}_l^\top A' \mb{v}_l
\end{array} \right), ~ A' = \partial_{\eta} A(\eta)|_{\eta=0}.
\end{align}
Then, the eigenvalue $\lambda_i(\eta)$ under the perturbation $\eta$ satisfies the following equality,
\begin{align} \lambda_i(\eta) = \lambda_i(0) + \eta \lambda^\mc{S}_i + \mc{O}(\eta).
\end{align}
\end{lem}

\begin{lem} \label{lem_dbound} \cite[Theorem~39.1]{bhatia2007perturbation}
Define matrix~${A_h(\eta) = A_h^0 +\eta A_h^1}$ of size $2np$, where $A_h^0$ and $A_h^1$ are $\eta$-independent. Define the optimal matching distance $d(\sigma(A_h),\sigma(A_h^0))$ as the maximum distance between the eigenspectrum of $A_h$ and $A_h^0$. Then,
\begin{align}
d(\sigma(A_h),\sigma(A_h^0))\leq 4(\lVert A_h^0\rVert+\lVert A_h\rVert)^{1-\frac{1}{np}} \lVert \eta A_h^1\rVert^{\frac{1}{np}},
\end{align}
for~${\min_{\pi} \max_{1\leq i\leq 2np} (\lambda_i - \lambda_{\pi(i)}(\eta))}$ with~$\pi(i)$ as the~$i$th permutation over $2np$ symbols.
\end{lem}

Next, for the proof of convergence to optimal state $\mb{x}=\mb{x}^*$ and $\mb{y}=\mb{0}_{np}$ under \textbf{NP-GT} algorithm, we prove that all the eigenvalues of $A_h$ except $p$ zero eigenvalues are in LHP. This implies the stability of consensus-based regime \eqref{eq_xdot_g}-\eqref{eq_ydot_g}. Recall that this follows from \cite{nonlin}, saying that the stability of the nonlinear dynamics is equivalent to the stability at all its operating points over time. Note that, unlike the convex optimization scenario in \cite{nedic2009distributed,dsvm,kar2012distributed,chen2012diffusion,mateos2014distributed,zhang2023accelerated,alghunaim2019distributed,ojsys,lian2023distributed}, we do not require strong-convexity condition $\nabla^2 f_i(\mb x) \succ 0$ at any node $i$.
\begin{theorem} \label{thm_zeroeig}
The $2np \times 2np$ matrix $A_h$ associated with the dynamics \eqref{eq_xydot1}-\eqref{eq_M_g} has only $p$ zero eigenvalues and the rest of eigenvalues are in the LHP for sufficiently small $\eta$.
\end{theorem}
\begin{proof}
Recall from the \textbf{NP-GT} dynamics~\eqref{eq_xdot_g}-\eqref{eq_ydot_g} that $A_h = A^0_h + \eta A^1$, and the eigen-spectrum satisfies Eq. \eqref{eq_spect_k}. Recall from Section~\ref{sec_prob} that matrix~$A^0$ has~$2p$ zero eigenvalues. For $j \in \{1,\dots,p\}$ associated with the $p$ dimensions of $\mb{x}_i$, other eigenvalues satisfy
$$\operatorname{Re}\{\lambda_{2n,j}\} \leq \ldots \leq \operatorname{Re}\{\lambda_{3,j}\} < \lambda_{2,j} = \lambda_{1,j} = 0.$$
To check how the eigenvalues (including the zero ones) will change by adding $\eta A_1$ as a perturbation, we use Lemma~\ref{lem_pert}. We specifically need to ensure that
no eigenvalue moves to the RHP, causing instability. Let~$\lambda_{1,j}(\eta,t)$ and~$\lambda_{2,j}(\eta,t)$ denote the perturbed zero eigenvalues. From Lemma~\ref{lem_pert}, we find the right and left eigenvectors associated with the zero eigenvalues as
\begin{align} \nonumber
\mb{v} = [\mb{v}_1~\mb{v}_2] =\frac{1}{\sqrt{n}} \left(\begin{array}{cc}
\mb{1}_n& \mb{0}_n \\
\mb{0}_n & \mb{1}_n
\end{array} \right)\otimes I_p,
\end{align}
and $\mb{u}=\mb{v}^\top$. This follows from the properties of the Laplacian matrix as explained in Section~\ref{sec_prob}. To simplify the derivation, we skip the normalization terms (due to the nonlinearity).
This does not change the outcome since we only need to know the \textit{sign} of multiplications not the exact value.
From \eqref{eq_beta_M}, we then have
\begin{align} \nonumber
\mb{v}_h = [\mb{v}_{h1}~\mb{v}_{h2}] = \left(\begin{array}{cc}
\mb{1}_n& \mb{0}_n \\
\mb{0}_n & \mb{1}_n
\end{array} \right)\otimes I_p.
\end{align}
In order to use Lemma~\ref{lem_pert},
one can find $\partial_{\eta} A_h(\eta)|_{\eta=0}=A_1$. Then,
\begin{eqnarray} \label{eq_dmalpha}
\mb{v}_h^\top A_1 \mb{v}_h= \left(\begin{array}{cc}
\mb{0}_{p\times p} & \mb{0}_{p\times p} \\
\dots & -(\mb{1}_n \otimes I_p)^\top H (\mb{1}_n \otimes I_p)
\end{array} \right).
\end{eqnarray}
Recall from Eq.~\eqref{eq_H1} that $
-(\mb{1}_n \otimes I_p)^\top H ( \mb{1}_n \otimes I_p) \prec 0
$. Therefore,
the matrix \eqref{eq_dmalpha} has $p$ zero and~$p$ negative eigenvalues associated with the zero eigenvalues of $A_h^0$, since~${\partial_{\eta} \lambda_{1,j}|_{\eta=0} = 0}$, ${\partial_{\eta} \lambda_{2,j}|_{\eta=0}<0}$. This implies that the perturbation~$\eta A_1$ puts the $p$ zero eigenvalues of~$A^0_h$ at LHP and other $p$ eigenvalues $\lambda_{1,j}(\eta,t)$ are still set to zero. Next, we need to guarantee that the rest of LHP eigenvalues of $A_h^0$ do not move to RHP due to perturbation. This is by recalling the matching distance in Lemma~\ref{lem_dbound} and is guaranteed by having
\begin{align} \label{eq_proof_d}
d(\sigma(A_h),\sigma(A^0_h)) < \kappa |\operatorname{Re}\{\lambda_{3,j}(\eta,t)\} |.
\end{align}
In other words, this ensures that the eigenvalues remain in LHP (i.e., ${\operatorname{Re}\{\lambda_{3,j}(\eta,t)\},\ldots,\operatorname{Re}\{\lambda_{2n,j}(\eta,t)\}}<0$ for $j \in \{1,\dots,p\}$).
Fig.~\ref{fig_dsigma} better illustrates the idea behind using Lemma~\ref{lem_dbound}.
\begin{figure}
\centering
\includegraphics[width=2.75in]{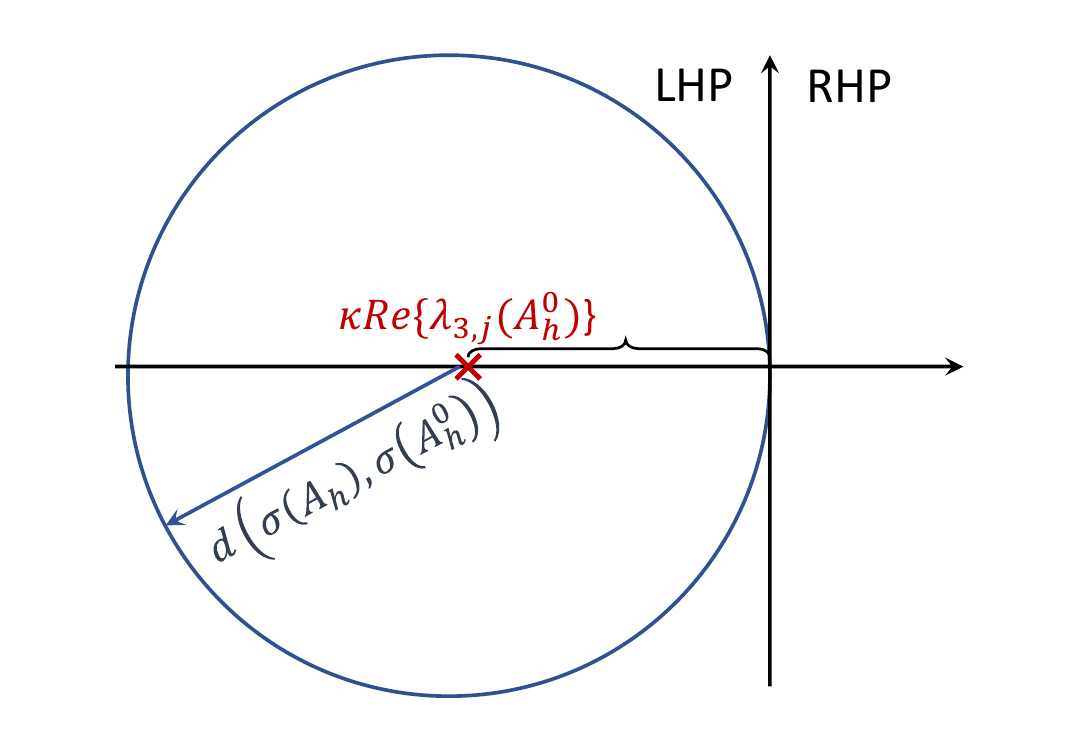} 
\caption{This figure shows the perturbation-based bound on the optimal matching distance between the eigenspectrum of $A_h$ and $A_h^0$. By having $d(\sigma(A_h),\sigma(A^0_h)) < \kappa |\operatorname{Re}\{\lambda_{3,j}(\eta,t)\}|$ the perturbation theory guarantees that the nonzero eigenvalues of $A_h$ remain in the LHP. } \label{fig_dsigma}
\end{figure}
Next, define ${\varphi := \max_{1 \leq i \leq np} \sum_{j=1}^{np} |H_{ij}|}$. Then, from \eqref{eq_sum_wij} and the definition of $\lVert\cdot\rVert_\infty$, $\lVert A_h^0 \rVert_\infty\leq 2\mc{K}(1+\varphi)$, $\lVert A^1 \rVert_\infty \leq \max\{ \varphi, 1\}$,
and thus, $\lVert A_h \rVert_\infty \leq \max\{2\mc{K}+\varphi(2\mc{K}+\eta), 2\mc{K}+\eta\}$.
Then, two cases can happen. For the first case, when $\varphi<1$ we define

\small
\begin{align} \nonumber
\overline{\eta} := \argmin_{\eta>0} \big( (2\mc{K}+2\varphi +\max&\{2\mc{K}+\varphi(2\mc{K}
+\eta), 2\mc{K}+\eta\})^{1-\frac{1}{np}} \\
&\eta^{\frac{1}{np}}-\kappa |\operatorname{Re}\{\lambda_{3,j}(\eta,t)\} |\big). \label{eq_alphabar1}
\end{align} \normalsize
Otherwise, when~$\varphi \geq1$ we define

\small 
\begin{equation} \label{eq_alphabar2}
\overline{\eta} = \argmin_{\eta>0} |4(4\mc{K}+\varphi (4\mc{K}+\eta))^{1-\frac{1}{np}} (\eta \varphi)^{\frac{1}{nm}}-\kappa |\operatorname{Re}\{\lambda_{3,j}(\eta,t)\} |.
\end{equation} \normalsize
Therefore, for any $\eta \leq \overline{\eta}$, Eq.~\eqref{eq_proof_d} holds. This ensures that all the system eigenvalues remain in the LHP.
\end{proof} 
Note that, in the proof analysis, we did not assume $\nabla^2 f_i(\mb x) \succ 0$, i.e., we made no assumption on the convexity of the local objectives. It should be mentioned that, in general, the eigenvalues may still remain in the LHP for a possibly less conservative choice of~$\eta>\overline{\eta}$. Further,
for sufficiently small~$\eta$, system dynamics~$A_h$ has~$p$ zero eigenvalues with eigenvector~$\mb{v}_{h1}$. Thus, the \textit{null space} of the \textit{time-varying}~$A_h$, defined by $ \text{span}\{[\mb{1}_n ;\mb{0}_n ]\otimes I_p\}$, is \textit{time-independent}. Using these facts along with the Lyapunov theorem, we prove the convergence to the optimizer in the next theorem.
\begin{theorem} \label{thm_lyapunov}
For sufficiently small $\eta$ satisfying Theorem~\ref{thm_zeroeig}, the proposed Algorithm~\ref{alg_1} converges to~$[\mb{x}^*;\mb{0}_{np}]$.
\end{theorem}
\begin{proof}
Define the PSD Lyapunov function ${V(\delta) = \frac{1}{2} \delta^\top \delta = \frac{1}{2}\lVert \delta \rVert_2^2}$ with the variable $\delta$ defined as
$$\delta = \left(\begin{array}{c} {\mb{x}} \\ {\mb{y}} \end{array} \right) - \left(\begin{array}{c} {\mb{x}}^* \\ \mb{0}_{np} \end{array} \right) \in \mathbb{R}^{2np}.$$
This variable denotes the difference between the system state and the optimal state and $\delta \rightarrow \mb{0}_{2np}$ implies that the state variables converge to the optimal state. First, recall that $[\mb{x}^*;\mb{0}_{np}]$ belongs to the null-space and is the invariant state of the \textbf{NP-GT} dynamics. We have,
\begin{align} \nonumber
\dot{\delta} &= \left(\begin{array}{c} \dot{\mb{x}} \\ \dot{\mb{y}} \end{array} \right) - \left(\begin{array}{c} \dot{\mb{x}}^* \\ \mb{0}_{np} \end{array} \right) \\&= A_h \left(\begin{array}{c} \mb{x} \\ \mb{y} \end{array} \right) - A_h\left(\begin{array}{c} \mb{x}^* \\ \mb{0}_{np} \end{array} \right) = A_h \delta.
\end{align}
Then, $\dot{V} = {\delta}^\top \dot{\delta}= \delta^\top A_h {\delta}$. Recall that, due to Theorem~\ref{thm_zeroeig}, the largest nonzero eigenvalue of $A_h$, denoted by $\lambda_{2,j}(\eta)$, is in the LHP. Then, using the results from~\cite[Sections~VIII-IX]{SensNets:Olfati04} on system theory, we have
\begin{eqnarray} \label{eq_Re2}
\dot{V} = \delta^\top A_h \delta \leq \max_{1\leq j\leq p}\operatorname{Re}\{{\lambda}_{2,j}(\eta)\} \delta^\top \delta,
\end{eqnarray}
which is negative-definite for $\delta \neq \mb{0}_{2np}$. Based on the Lyapunov theorem, we can complete the proof.
\end{proof}
The above proof further shows that the convergence rate depends on the largest non-zero eigenvalue of $A_h$. Recalling Eq.~\eqref{eq_M_g}, the eigenspectra of $A_h$ depends on the eigenspectrum of the Laplacian matrix $\overline{W}_\gamma$. It is known that the largest nonzero eigenvalue of $\overline{W}_\gamma$ (referred to as the \textit{algebraic connectivity}) directly determines the convergence rate of the consensus-based dynamics \cite{SensNets:Olfati04}. This implies that for networks with higher connectivity, the convergence rate is faster. This is better illustrated in the next section by simulations.

\begin{rem}
The perturbation-based analysis in this section allows to address change in the eigenspectrum of the proposed dynamics while satisfying system stability. In other words, the proposed solution converges over any time-varying network topology as far as its perturbed eigenvalues remain in the LHP, satisfying system stability. This is a key feature of the proposed \textbf{NP-GT} algorithm that allows to address time-varying network topologies and convergence in the presence of link failure. This makes the \textbf{NP-GT} algorithm superior to other existing GT-based solutions in the literature.
\end{rem}

\begin{rem}
Recall that, although the local cost functions $f_i(\cdot)$ are non-convex and may have local minimum, the gradient-based dynamics moves from these local minimums and converges to the global minimum of the global cost function $F(\cdot)$. This follows from the structure of the proposed coupled dynamics, in which both the gradient-tracking in Eq.~\eqref{eq_ydot_g} and state-update in Eq.~\eqref{eq_xdot_g} are based on averaging over the neighboring nodes. Adopting this cooperative strategy, the consensus-coupling structure of the problem~\eqref{eq_prob} under non-convex Assumption~\ref{ass_cost} is addressable. An example of such non-convex setup is given later in Section~\ref{sec_sim_nonconv}.
\end{rem}

\section{Simulations}\label{sec_sim}
\subsection{Convex Linear Regression Problem}
For the first simulation, we consider the linear regression problem (with quadratic objective function) to be compared with some existing literature. In this setup, each machine/node optimizes the following objective function,
\begin{equation} \label{eq_fi}
f_i(\boldsymbol{\beta}_i,\nu_i) = \frac{1}{m_i} \sum_{j=1}^{m_i} (\boldsymbol{\beta}_i^\top \boldsymbol{\chi}_j^i - \nu_i -y_j)^2,
\end{equation}
with $\boldsymbol{\beta}_i \in \mathbb{R}^2$ and $\nu_i \in \mathbb{R}$ denoting the regressor line parameters fitting the local data. Each machine $i$ is assigned a batch of $m_i$ data points with coordinates $\boldsymbol{\chi}_j^i$ and $y_i$ in 2D space. The goal is to find the linear regressor optimally fitting the data points locally and distributedly. In the distributed formulation \eqref{eq_prob}-\eqref{eq_fij}, the state variable is $\mb{x}_i = [\boldsymbol{\beta}_i ; \nu_i]$ and all the machines/nodes need to reach consensus on this regression parameters. We consider $n=10$ machines with access to $m_i=50$ randomly chosen data points (out of $m=100$ points). Each machine optimizes its regression parameters for its batch of data and shares gradient information over a random Erdos-Renyi (ER) network with linking probability $25\%$. The nonlinearity of the links is considered log-scale quantization. The regressor parameters at different nodes with respect to time are shown in Fig.~\ref{fig_beta}. For this simulattion we set $\rho = \frac{1}{256}$ and $\eta = 2$.
\begin{figure}
\centering
\includegraphics[width=1.7in]{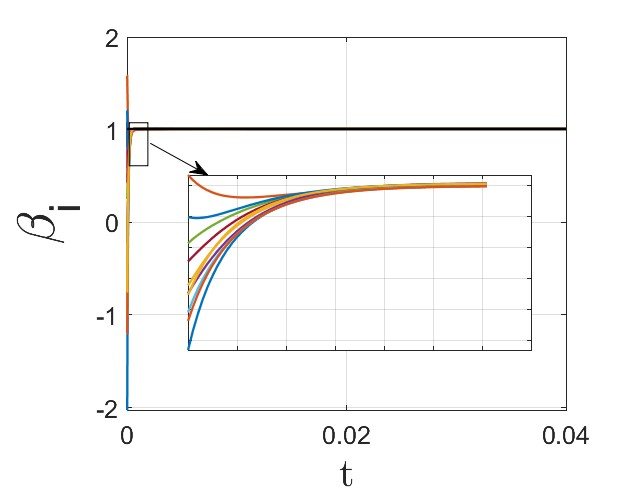}
\includegraphics[width=1.7in]{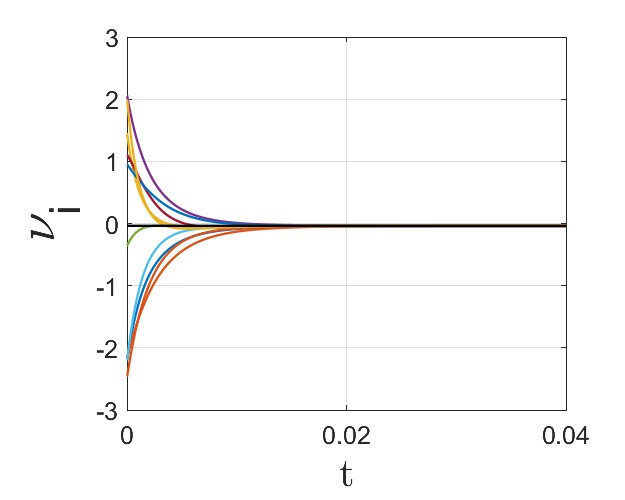} 
\caption{This figure shows the convergence of the local regressor parameters to the optimal centralized linear regressor. As illustrated, the agents reach consensus on the regression parameters $\beta_i,\nu_i$. } \label{fig_beta}
\end{figure}

We compare our solution subject to log-quantized data exchange with linear \cite{xin2021fast}, fixed-time \cite{garg2020fixed1}, and sign-based \cite{taes} continuous-time protocols given in the literature. All these works are based on gradient-tracking (GT) without considering link nonlinearity. Since these works are not based on perturbation analysis, they do not necessarily converge over time-varying networks, and thus, the network in this case is considered time-invariant. The residual (optimality gap) is compared in Fig.
\ref{fig_compare}. In this simulation, we consider heterogeneous data among different nodes, i.e., every node has access only to its batch of partial data and \textit{not} all the data points. The simulation shows that the proposed dynamics are robust to this data heterogeneity and converge in the face of this partial data localization.
\begin{figure}
\centering
\includegraphics[width=2.75in]{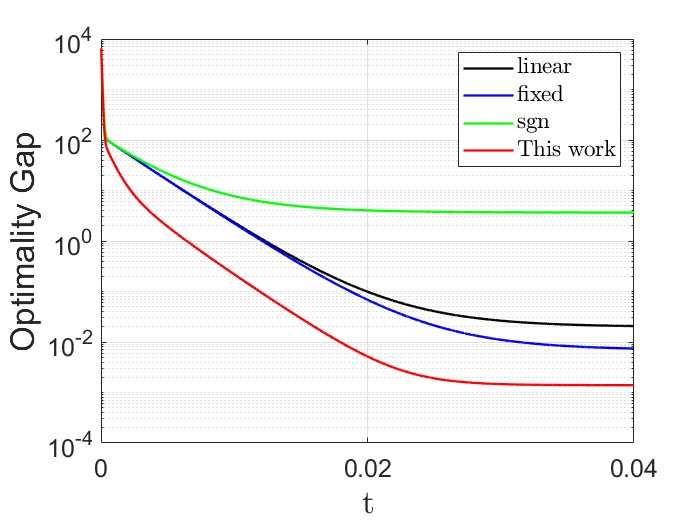} 
\caption{This figure compares the optimality gap of the objective function~\eqref{eq_fi} (linear regression) under different distributed optimization techniques. Assuming heterogeneous setups, each agent can access its \textit{local batch of data}. } \label{fig_compare}
\end{figure}

\subsection{Non-Convex Problem} \label{sec_sim_nonconv}
\begin{figure*}
\centering
\includegraphics[width=1.75in]{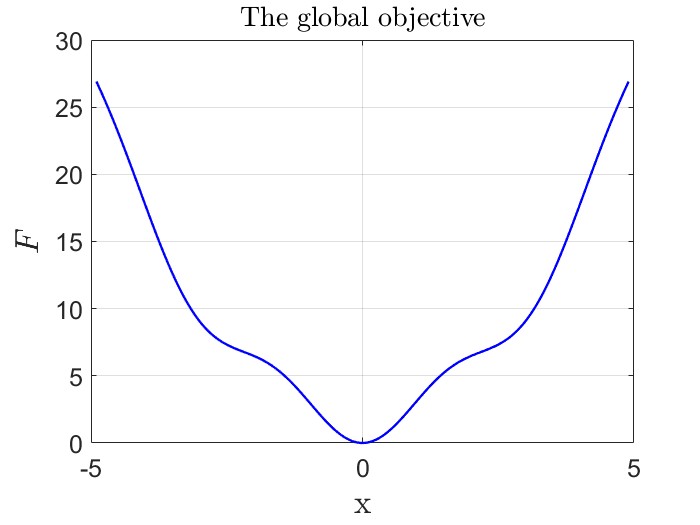}
\includegraphics[width=1.75in]{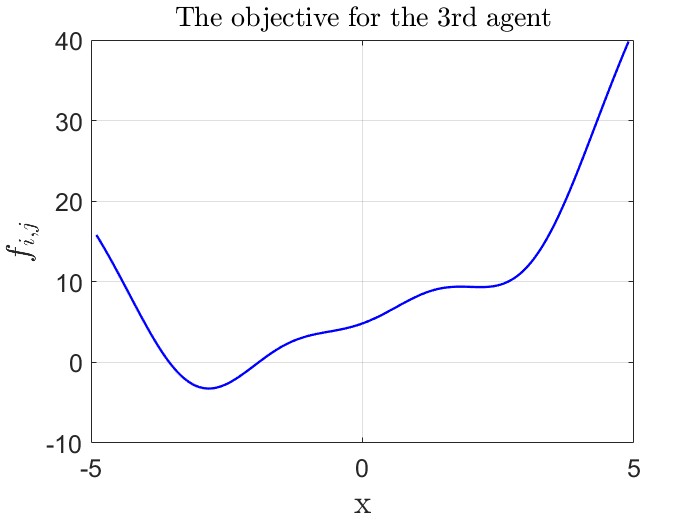}
\includegraphics[width=1.75in]{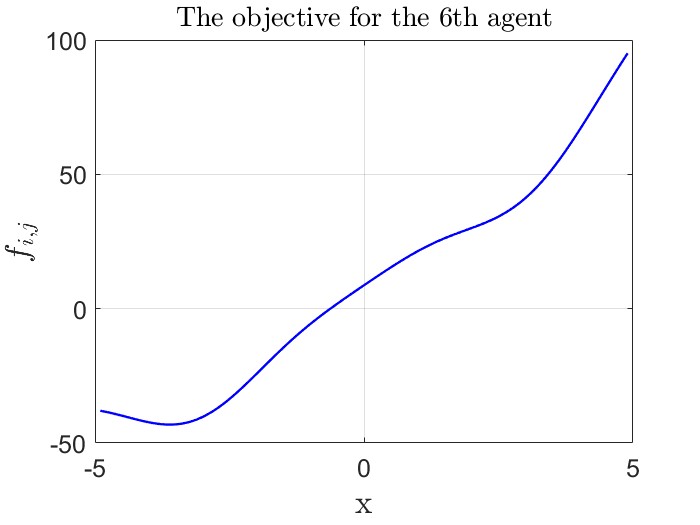}
\includegraphics[width=1.75in]{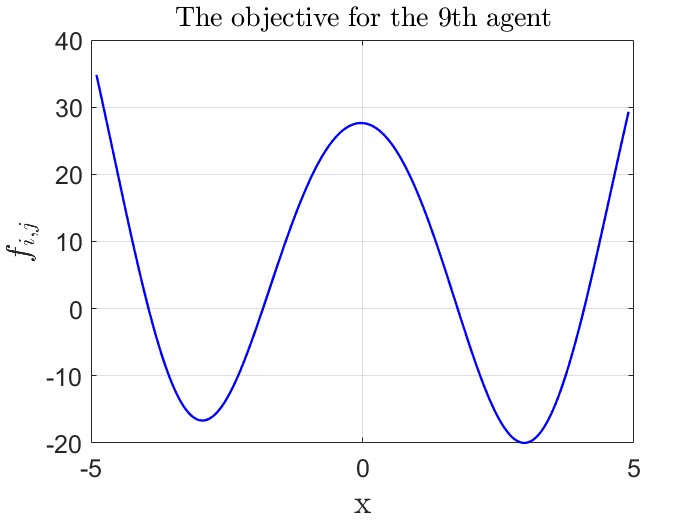} 
\caption{The Left figure shows the global cost function $F(\mb x) = \frac{1}{n}\sum_{i=1}^{n} \frac{1}{m}\sum_{j=1}^{m} f_{i,j}(\mb{x}_i)$. The other three figures show the local non-convex objective functions at three sample nodes. } \label{fig_nonconv}
\end{figure*}
Next, we consider a synthetic non-convex local objective function to be optimized. The simulation is performed over a network of $n=10$ agents/nodes and $m=40$ sample data points.
The objective is defined similarly to the one from \cite{xin2021fast}:
\begin{align}\label{eq_fij_sim}
f_{i,j}(x_i) = 2 x_i^2 +3\sin^2(x_i)+a_{i,j} \cos(x_i) + b_{i,j}x_i,
\end{align}
with $\sum_{i=1}^n \sum_{j=1}^m a_{i,j} = 0$ and $\sum_{i=1}^n \sum_{j=1}^m b_{i,j}=0$ such that $a_{i,j},b_{i,j} \neq 0$ and set randomly in the range $(-5,5)$. Note that these local cost functions are not convex, i.e., $\nabla^2 f_{i}(\mb{x})$ might be negative at some points. This is illustrated in Fig.~\ref{fig_nonconv}.
The agents' communication network is considered an ER random graph with linking probability $20\%$. This graph topology randomly \textit{changes} at every $1$ sec. We set the GT tracking parameter as $\eta = 1$.
Considering log-scale quantization nonlinearity at the links, we repeat the simulations for three different values of quantization level: $\rho=\frac{1}{128}$, $\rho=\frac{1}{512}$, $\rho=\frac{1}{1024}$. The simulation is shown in Fig.~\ref{fig_nonconv2}.
\begin{figure}
\centering
\includegraphics[width=2.75in]{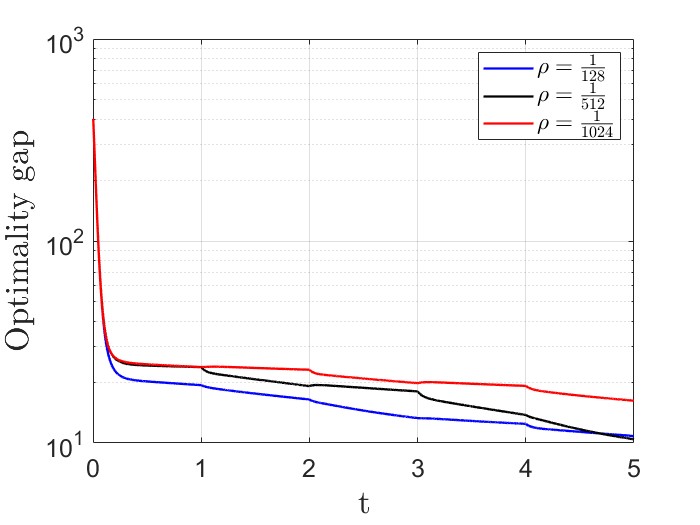} 
\caption{This figure shows the time-evolution of the locally non-convex cost-function~\eqref{eq_fij_sim} subject to log-scale quantized links and time-varying network topology. The change in the network topology causes the non-smoothness of the residual (optimality gap) evolution. } \label{fig_nonconv2}
\end{figure}
The non-smooth time evolution of the optimality gap is because of the change in the network topology. We repeat the simulation for an \textit{n-hop exponential} network topology\footnote{In an exponential network of size $n$, every node is connected to its $2^0,2^1,\dots,2^{\lfloor \log_2(n-1) \rfloor}$-hop neighbors (with $\lfloor \cdot \rfloor$ denoting the greatest integer less than or equal to the number), see more details in \cite{assran2019stochastic}.}. It is known that this type of network topology shows linear convergence behavior with a low optimality gap for distributed learning and optimization \cite{assran2019stochastic,nedic2018network}. This is better illustrated in Fig.~\ref{fig_nonconv2exp}. For this simulation, we set the linking probability of the ER network as $50\%$, $\eta=2$, and $\rho = \frac{1}{64}$.
\begin{figure}
\centering
\includegraphics[width=2.75in]{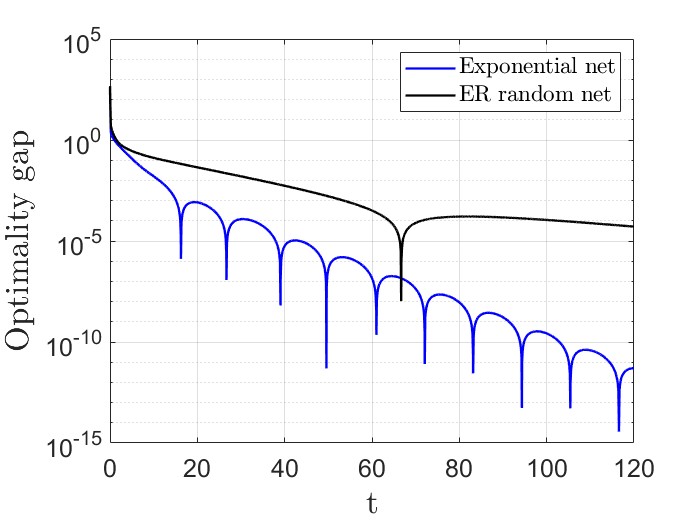} 
\caption{This figure compares the optimality gap for exponential network versus random ER network. As it is clear from the figure, for the same objective~\eqref{eq_fij_sim} and log-scale quantization level, the exponential network reaches a lower residual and optimality gap. } \label{fig_nonconv2exp}
\end{figure}

\subsection{Link Failure Scenario}
We redo the simulation for the same objective and parameters.
Consider the case where some links over the network fail. That may represent packet drops, for example. Since the proposed \textbf{NP-GT} dynamics is defined over balanced networks and works over possibly switching and time-varying topology, it can handle link failure scenarios. Many existing works in the literature cannot address such scenarios; for example, the works \cite{xin2022fast,xin2021fast,qureshi2023distributed} can only handle static networks. To illustrate this, we consider an ER network with $30\%$ connectivity for this simulation and gradually remove some links while the network preserves its strong connectivity. We perform the simulations under different link-removal rates $p$ as shown in Fig.~\ref{fig_nonconv3}.
\begin{figure}
\centering
\includegraphics[width=2.75in]{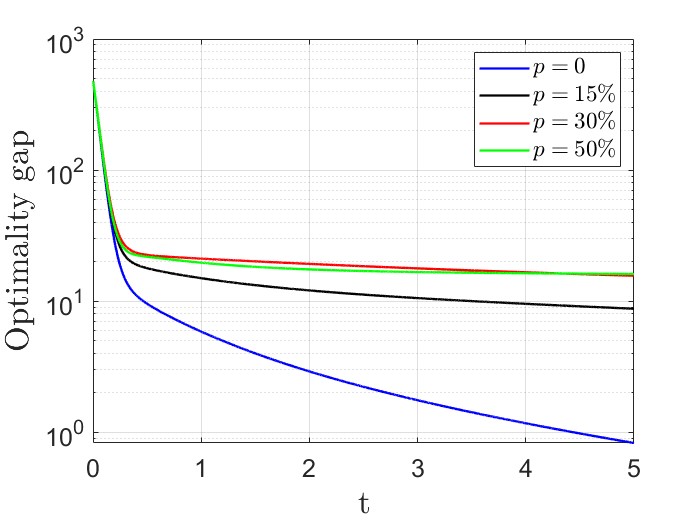} 
\caption{This figure shows the time-evolution of the locally non-convex cost-function~\eqref{eq_fij_sim} over log-scale quantized links and time-varying network topology under link failure rate $p$. Lower network connectivity due to the link failure results in a slower convergence rate.} \label{fig_nonconv3}
\end{figure}
As is evident from the figure, by removing more links, the convergence rate decreases. This is because the convergence rate depends on the actual value of the largest nonzero eigenvalue of $A_h$, which in turn is a function of the algebraic connectivity of the Laplacian matrix $\overline{W}_\gamma$ and according to \cite{SensNets:Olfati04} directly depends on the network connectivity.
\section{Conclusions}\label{sec_con}
\subsection{Concluding Remarks}
This work proposes decentralized algorithms to solve locally non-convex optimization functions. Our solution addresses optimization dynamics subject to (possible) nonlinear data transmission (for example, due to log-scale quantization) over multi-agent networks. We prove \textit{exact} convergence under odd sign-preserving sector-bound nonlinear setups using perturbation-based eigenspectrum analysis. To our knowledge, this is not addressed in the existing literature.
We further address the convergence under possible time variation of the network topology (for example, due to link failure), which results in a non-smooth decrease of the optimality gap over time. Our results advance the state-of-the-art distributed algorithms for machine learning and optimization in terms of real-world networking constraints and the practical challenges of multi-agent processing.

\subsection{Future Directions}
Future research efforts will focus on different machine learning techniques to better model and mitigate the effects of various link nonlinearities, such as convergence under uniform quantization, which is not sector-bound. Applying the proposed methodology to data-driven approaches and decentralized swarm robotic networks (for example, assigning optimal formation setups \cite{chen2023distributed,tase}) opens new insights that may overcome traditional existing challenges. Further, exploring more general non-convex optimization solutions in the presence of link constraints is an ongoing research frontier.

\section*{Acknowledgements}
This work has been supported by the Center for International Scientific Studies \& Collaborations (CISSC), Ministry of Science Research and Technology of Iran.

\bibliographystyle{IEEEbib}
\bibliography{bibliography}

\begin{IEEEbiography}[{\includegraphics[width=1.1in,clip,keepaspectratio]{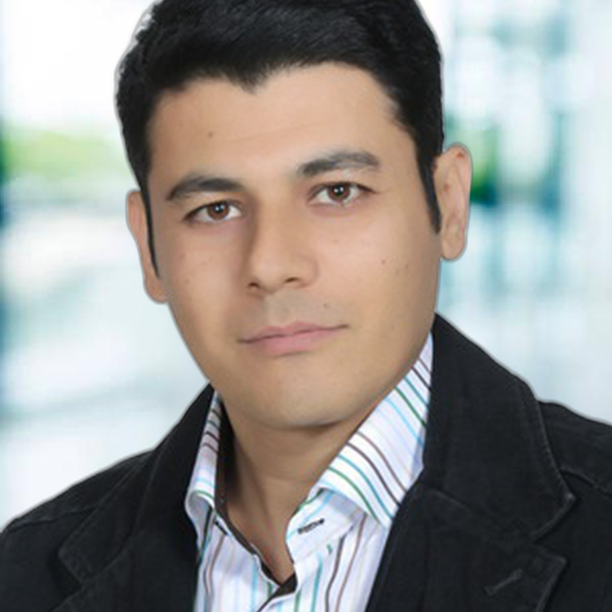}}]{Mohammadreza~Doostmohammadian}
He received his B.Sc. and M.Sc. in Mechanical Engineering from Sharif University of Technology, Iran, respectively, in 2007 and 2010, where he worked on control systems and robotics applications. He received his PhD in Electrical and Computer Engineering from Tufts University, MA, USA 2015. During his PhD in Signal Processing and Robotic Networks (SPARTN) lab, he worked on control and signal processing over networks with applications in social networks. From 2015 to 2017, he was a postdoc researcher at the ICT Innovation Center for Advanced Information and Communication Technology (AICT), School of Computer Engineering, Sharif University of Technology, with research on network epidemic, distributed algorithms, and complexity analysis of distributed estimation methods. He was a researcher at the Iran Telecommunication Research Center (ITRC), Tehran, Iran, in 2017, working on distributed control algorithms and estimation over IoT. Since 2017, he has been an assistant professor in the Mechatronics Department at Semnan University, Iran. He was a visiting researcher at the School of Electrical Engineering and Automation, Aalto University, Espoo, Finland, working on constrained and unconstrained distributed optimization techniques and their applications. His general research interests include distributed estimation, control, learning, and network optimization. He was the chair of the robotics and control session at the ISME-2018 conference, the session chair at the 1st Artificial Intelligent Systems Conference of Iran, 2022, and the IPC member and Associate Editor of the IEEE/IFAC CoDIT2024 conference.
\end{IEEEbiography}

\begin{IEEEbiography}[{\includegraphics[width=1.05in,clip,keepaspectratio]{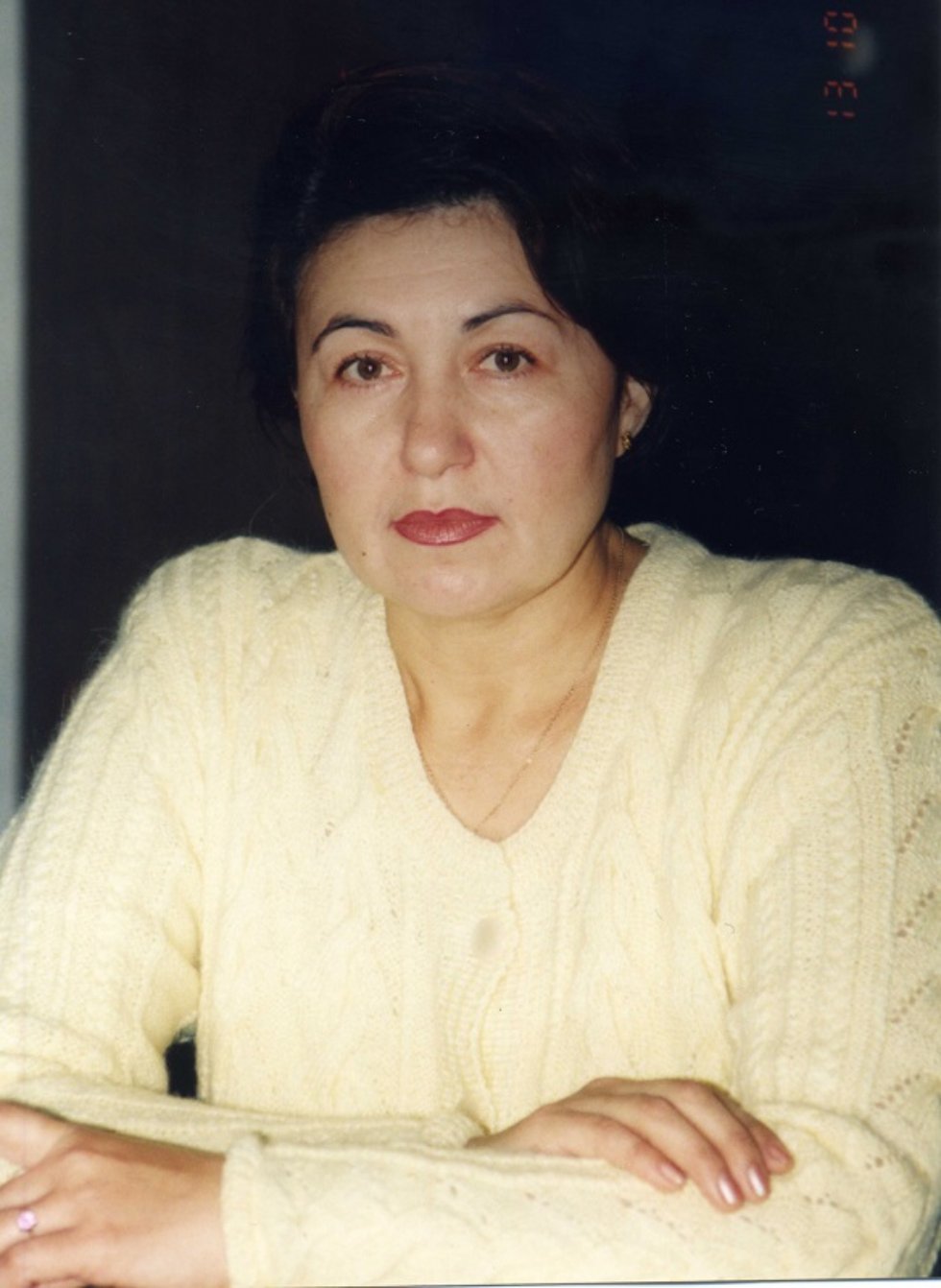}}]{Zulfiya~R.~Gabidullina} graduated with honours from the Faculty of
Computational Mathematics of the Kazan State University (Russia) on a specialty “Economic Cybernetics” in 1982 and received a qualification: Economist-Mathematician. From 1982 to 1984, she worked as an engineer-mathematician at the Computer Center of Kazan State University. From 1984 to 1987, she received an After-college education at Kazan State University, specializing in “Computational Mathematics”. From 1987 to 1993 and from 1993
to 1996, she was an Assistant and Senior Lecturer
(Department of Economic Cybernetics at the Faculty of Computational Mathematics and Cybernetics of Kazan State University). In 1994, she defended the candidate's dissertation and received the Candidate of Physics-Mathematical Sciences degree. Since 1997, she has had the academic title of Associate Professor. Since 1996, she has worked as an Associate professor at the Department of Data Analysis and Programming Technologies, Institute of Computational Mathematics and Information Technologies, Kazan Federal University, Kazan, Russia. She has taken internships at Moscow State University (Russia) and DePaul University (Chicago, Illinois, USA).
She was the Session Chair: The 21st International Symposium on Mathematical Programming, Berlin, Germany, ISMP-2012.
Her main research interests include optimization methods and their applications, machine learning, and intellectual data analysis.
\end{IEEEbiography}

\begin{IEEEbiography}[{\includegraphics[width=1.1in]{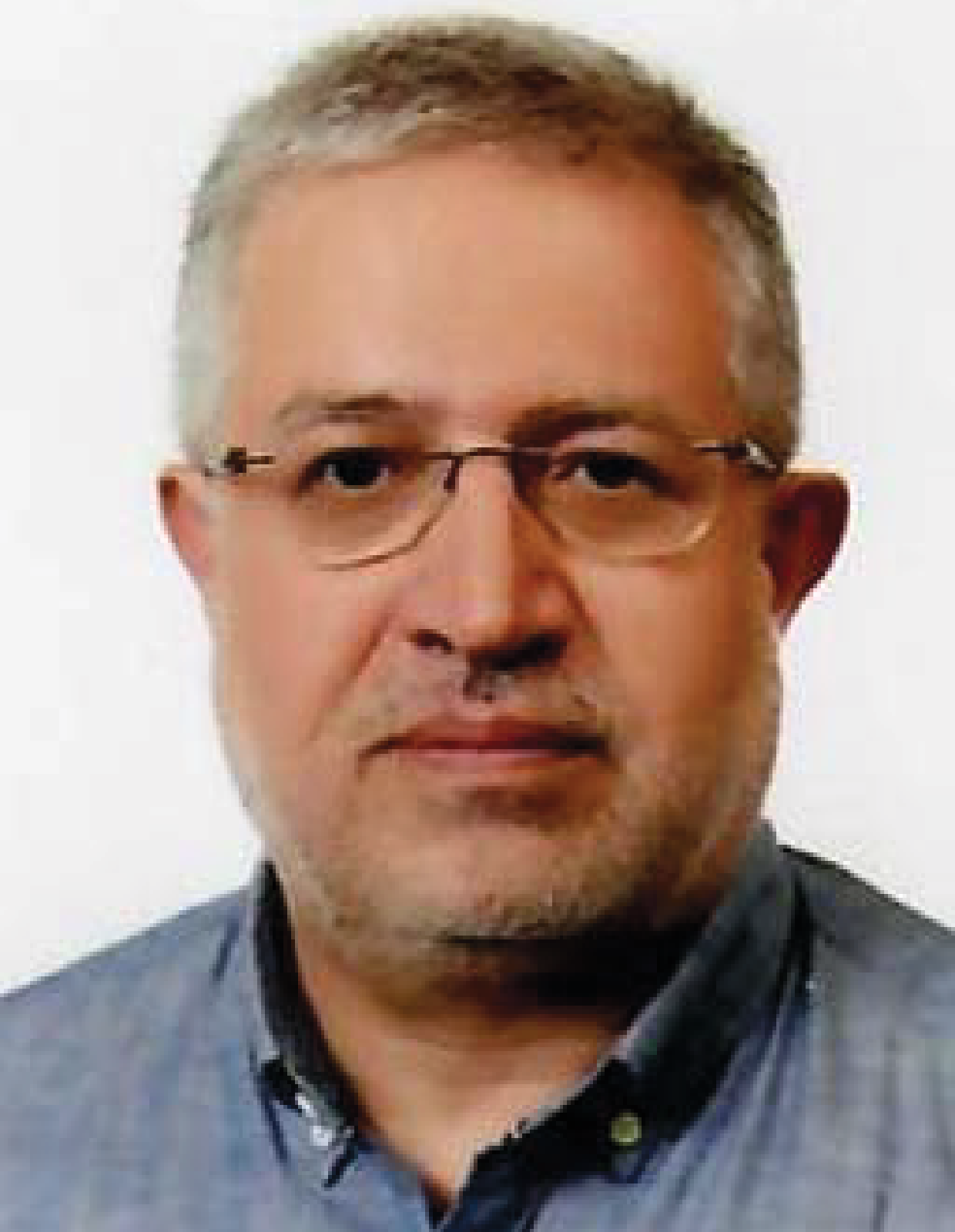}}]{Hamid R. Rabiee}
	received his BS and MS degrees (with Great Distinction) in Electrical Engineering from CSULB, Long Beach, CA (1987, 1989), his EEE degree in Electrical and Computer Engineering from USC, Los Angeles, CA (1993), and his Ph.D. in Electrical and Computer Engineering from Purdue University, West Lafayette, IN, in 1996. From 1993 to 1996, he was a Member of Technical Staff at AT\&T Bell Laboratories. From 1996 to 1999, he worked as a Senior Software Engineer at Intel Corporation. He was also with PSU, OGI and OSU universities as an adjunct professor of Electrical and Computer Engineering from 1996-2000. Since September 2000, he has joined the Sharif University of Technology, Tehran, Iran. He was also a visiting professor at the Imperial College of London for the 2017-2018 academic year. He is the founder of Sharif University Advanced Information and Communication Technology Research Institute (AICT), ICT Innovation Center, Advanced Technologies Incubator (SATI), Digital Media Laboratory (DML), Mobile Value Added Services Laboratory (VASL), Bioinformatics and Computational Biology Laboratory (BCB) and Cognitive Neuroengineering Research Center. He has also been the founder of many successful High-Tech start-up companies in the field of ICT as an entrepreneur. He is currently a Professor of Computer Engineering at Sharif University of Technology and Director of AICT, DML, and VASL. He has been the initiator and director of many national and international level projects in the context of Iran's National ICT Development Plan and UNDP International Open Source Network (IOSN). He has received numerous awards and honors for his Industrial, scientific, and academic contributions. He has acted as chairman in a number of national and international conferences and holds three patents. He is also a Member of IFIP Working Group 10.3 on Concurrent Systems and a Senior Member of IEEE. His research interests include statistical machine learning, Bayesian statistics, data analytics, and complex networks with applications in social networks, multimedia systems, cloud and IoT privacy, bioinformatics, and brain networks.
\end{IEEEbiography}

\end{document}